\documentclass[letterpaper, 10 pt, conference]{ieeeconf}

\usepackage[
 letterpaper,
 left=48pt,
 right=48pt,
 top=60pt,
 bottom=43.5pt
]{geometry}

\IEEEoverridecommandlockouts

\usepackage{amsmath, amsthm}
\usepackage{amssymb}
\usepackage{booktabs}
\usepackage{graphicx}
\usepackage{cite}
\usepackage{algorithm}
\usepackage{algorithmic}
\usepackage{xcolor}

\newtheorem{theorem}{Theorem}

\newtheorem{lemma}{Lemma}
\newtheorem{corollary}{Corollary}
\newtheorem{assumption}{Assumption}

\DeclareMathOperator*{\minimize}{minimize}

\newcommand{\Aset}{\mathcal A}
\newcommand{\Jset}{\mathcal J}
\newcommand{\Cset}{\mathcal C}
\newcommand{\Xset}{\mathcal X}
\newcommand{\Uset}{\mathcal U}
\newcommand{\R}{\mathbb R}
\newcommand{\Prob}{\mathbb P}
\newcommand{\E}{\mathbb E}
\newcommand{\1}{\mathbf 1}
\newcommand{\norm}[1]{\left\lVert #1 \right\rVert}
\newcommand{\Pred}{\mathrm{Pred}}
\newcommand{\st}{\text{s.t.}}
\newcommand{\pospart}[1]{\left[#1\right]_{+}}
\newcommand{\eps}{\varepsilon}

\title{\LARGE \bf Multi-Agent Deferred-Decision Trajectory Optimization with Anytime-Valid Target Retention for Moving Candidates}

\author{Farbod Siahkali$^{1}$, Mengxue Hou$^{2}$, and Vijay Gupta$^{1}$%
\thanks{This work was supported in part by the U.S. Army Research Office under Grant 13001664 and in part by the Air Force Office of Scientific Research (AFOSR) under DEPSCoR Award FA9550-25-1-0283.}
\thanks{$^{1}$Farbod Siahkali and Vijay Gupta are with the Department of Electrical and Computer Engineering, Purdue University, West Lafayette, IN 47906, USA (Emails: \texttt{siahkali@purdue.edu}, \texttt{gupta869@purdue.edu}).}%
\thanks{$^{2}$Mengxue Hou is with the Department of Electrical Engineering, University of Notre Dame, Notre Dame, IN 46556, USA (Email: \texttt{mhou@nd.edu}).}%
}

\begin{document}
\maketitle

\begin{abstract}
We study multi-agent terminal interception when one of several moving candidates is the unknown target and both target identity and candidate terminal locations are uncertain. We develop a receding-horizon deferred-decision trajectory optimization (DDTO) controller that combines conformal terminal-location prediction, sequential candidate rejection, capacity-constrained assignment, and candidate-contingent planning. Under exchangeability conditions for terminal and identity residuals, with predictable querying, we establish fixed-time terminal-location coverage and anytime-valid target retention. At the final replan, when one control step remains, we bound the joint probability that a tightened-set DDTO plan passes nonlinear rollout verification but capture fails. The closed-loop controller applies fallback guidance to any agent without a verified DDTO plan and, across 100 paired scenarios, captures $79\%$ of targets under the strict control-effort budget and $85\%$ when that budget is relaxed.
\end{abstract}


\section{Introduction}
\label{sec:intro}

Multi-robot autonomy naturally couples sensing and control, since motion affects what information can be collected and sensing, in turn, affects subsequent decisions~\cite{norouzi2025novel,lee2025dmvc,yuan2025hierarchical,rezaei2025online}. We consider a team of agents tasked with intercepting one unknown target among several moving candidates, where this coupling arises through both target identification and motion planning. Neither the target identity nor the candidates' terminal locations are known when motion begins.

The agents must therefore make progress toward interception while preserving alternatives that may later become necessary. Committing too early can direct an agent toward a non-target, whereas excessive deferral can leave insufficient time or control authority for capture. The central difficulty is that inference, assignment, and trajectory planning are coupled~\cite{khass2026multi}. Because target identity is unresolved when the current control is applied, each agent must choose a single input that preserves feasible trajectories toward all assigned candidates.

\begin{figure*}[t]
\centering
\includegraphics[width=0.9\textwidth]{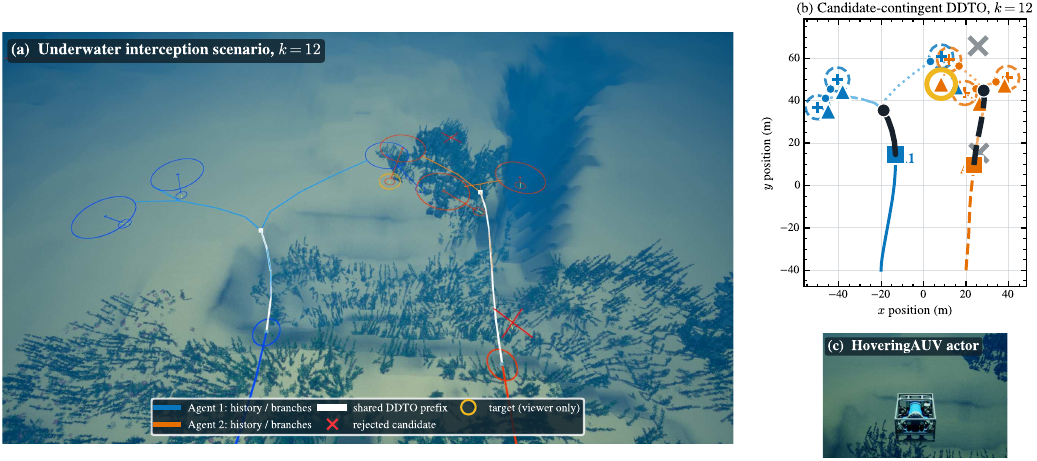}
\caption{Closed-loop HoloOcean demonstration of multi-AUV terminal interception at replan $k=12$.
(a) HoloOcean state after executing the control applied at the preceding replan.
(b) Corresponding planner state, initialized from simulator feedback, showing assignments, predicted terminal centers, conformal terminal regions, common DDTO prefixes, and candidate-contingent branches.
(c) HoveringAUV model used for the agents.
The target is highlighted in the figure but is not available to the controller.}
\label{fig:framework}
\end{figure*}

Contingency and scenario-based planners maintain plans for multiple future realizations~\cite{peters2024contingency,de2021scenario,calafiore2012robust}. Contingency model predictive control (CMPC) shares the currently applied input across nominal and contingency plans~\cite{cmpc}, while delayed-decision motion planning (DDMP) preserves multiple feasible responses up to a prescribed delayed decision time~\cite{ddmp}. DDTO constructs candidate-contingent branches with a common prefix and delays their separation while preserving terminal feasibility~\cite{elango2022deferring,elango2025deferred}.

Existing DDTO formulations typically assume fixed candidate terminal locations. Prior active-sensing DDTO~\cite{siahkali2026active}, which incorporates sensing into the optimization, still considers one agent, fixed terminal points, and candidate-set guarantees assessed at individual times. It leaves terminal-location uncertainty for moving candidates, anytime-valid target retention under adaptive querying, and multi-agent assignment unresolved. We consider these issues jointly in a receding-horizon multi-agent setting.

Statistical and learning-based methods have been developed across both estimation and prediction settings~\cite{mesforush2026,strawn2023conformal,mazloom2026barl}. Conformal prediction (CP) provides model-agnostic uncertainty quantification for regression and classification problems~\cite{angelopoulos2023gentle,fusioncp}. The two uncertainties considered here require different statistical guarantees. For terminal-location uncertainty, CP can provide marginal coverage of a candidate's terminal location at any prespecified replanning index~\cite{lindemann2023safe,sheng2024safe,dixit2023adaptive,muthali2023reachcal}.

Identity decisions, however, are made repeatedly: at each replan, agents select which candidates to query using current tracking information and controller-visible history, before observing the corresponding identity measurements. Rejected candidates are not reconsidered. Consequently, applying a fixed-time test independently at each replan does not control the probability of ever rejecting the target. The identity procedure must instead provide an anytime-valid retention guarantee under predictable querying, meaning that each query is selected from information available before the queried identity measurement is observed.

We combine these components in a closed-loop multi-agent DDTO framework. At each replan, conformal terminal regions represent candidate terminal-location uncertainty. Candidate-wise online conformal p-values~\cite{fedorova2012plugin} are converted to e-values and accumulated in an e-process~\cite{vovk2021evalues}. Threshold crossings then generate nested retained sets. With the target's adjusted identity residuals conditionally exchangeable given the exogenous scene and other-candidate information, and with predictable querying, Ville's inequality bounds the probability that the target is ever rejected~\cite{ville1939etude}.

The retained candidates are assigned subject to per-agent capacity, and each agent solves a DDTO problem whose branches share the applied control. At the final replan, when one control step remains, the controller uses a tightened terminal set. Target retention and terminal-location coverage then bound the probability that a verified DDTO plan is returned at the final replan but capture fails. When an agent has no verified plan, the controller instead applies fallback guidance toward a retained candidate.

The contributions are threefold. First, we formulate a receding-horizon multi-agent interception architecture integrating terminal-location uncertainty for moving candidates, sequential candidate rejection, capacity-constrained assignment, and candidate-contingent control. Second, we integrate candidate-wise online conformal evidence with predictable querying and establish anytime-valid target retention under explicit information assumptions. Third, we combine fixed-time terminal-location coverage and anytime-valid target retention to bound the probability that a verified tightened-set plan is returned at the final replan but capture fails.

\section{Problem Formulation}
\label{sec:problem}

Let $\Aset=\{1,\ldots,M\}$ denote the agent set and $\Jset=\{1,\ldots,N\}$ the moving-candidate set. Exactly one candidate is the target, with unknown index $J^\star\in\Jset$. Agent $a$ has assignment capacity $\kappa_a$, and we assume $N \leq \sum_{a\in\Aset}\kappa_a$. We discretize the mission into decision stages $t=0,\ldots,T$, where $T$ is a prescribed terminal interception deadline. At replan $k$, $H_k=T-k$ denotes the remaining planning horizon. Agent $a\in\Aset$ has state $x_t^a\in\Xset_a\subseteq\R^{n_a}$, control input $u_t^a\in\Uset_a\subseteq\R^{m_a}$, and deterministic dynamics
\[
x_{t+1}^a=f_a(x_t^a,u_t^a).
\]
Matrix $P_a\in\R^{p\times n_a}$ extracts the agent position. Each agent satisfies the effort budget $\sum_{t=0}^{T-1}c_a(x_t^a,u_t^a)\le B_a$, where $c_a$ is the stage effort and $B_a>0$ is the available budget.

Let $y_t^j \in \mathbb{R}^p$ be the true position of candidate $j \in \mathcal{J}$, whose motion is exogenous to the agent controls. We assume a centralized controller: at each replan, the current agent states and candidate tracking histories are available to the controller. For candidate $j$, $Y_{0:k}^j = (\widetilde{y}_0^j,\ldots,\widetilde{y}_k^j)$ denotes noisy tracking observations available at the replanning rate and used for terminal prediction, while $y_T^j$ remains unknown for $k<T$.

Identity information is obtained through active queries. When agent $a$ queries candidate $j$, it receives the candidate-specific measurement $m_k^{a,j}$, which the identity model uses to assess whether $j$ is the target. Tracking observations are available at every replan, whereas identity measurements are allocated through the query policy.

\paragraph*{\textbf{Terminal-location uncertainty}} For capture radius $\rho>0$, define the terminal mission-success event
\[
\mathrm{Cap}:=\left\{\min_{a\in\Aset}\norm{P_a x_T^a-y_T^{J^\star}}_2\le\rho\right\}.
\]
We use $\mathrm{Cap}$ as the mission-success criterion in both the theoretical analysis and the closed-loop numerical evaluation. The objective is to generate causal controls for terminal capture while satisfying the dynamics, state, input, effort, assignment, and sensing constraints. From $Y_{0:k}^j$, the controller constructs a conformal terminal region $\mathcal Z_k^j$ for each candidate. For every fixed deterministic $k<T$, the terminal-prediction objective is
\[
\Prob\!\left(
y_T^{J^\star}\in\mathcal Z_k^{J^\star}
\right)\ge1-\alpha_{\rm tr}.
\]

\paragraph*{\textbf{Target-identity uncertainty}} Identity information is updated through active queries. Let $\Cset_k$ denote the candidate set retained immediately before the replan-$k$ identity query, and let $\widehat\Cset_k=\Cset_{k+1}$ denote the post-query retained set used for planning at the same replan. The identity objective is the anytime statement
\[
\Prob\!\left(J^\star\in\Cset_k\ \text{for all }k=0,\ldots,T\right)\ge1-\alpha_{\rm id}.
\]
After the identity update, each candidate in $\widehat\Cset_k$ is assigned exactly once subject to capacity $\kappa_a$. For an assigned pair $(a,j)$, the branch $\{x_{h|k}^{a,j},u_{h|k}^{a,j}\}$ starts from $x_k^a$ and must satisfy the dynamics, constraints, remaining budget, and terminal condition induced by $\mathcal Z_k^j$. Because target identity is unresolved when the current control is applied, all branches assigned to one agent share their first input.

\section{Proposed Method}
\label{sec:method}

At each replan, the controller constructs terminal regions, updates identity evidence, assigns retained candidates, and solves one DDTO problem per agent. A plan is verified only if it passes the solver checks and a nonlinear rollout under the planning dynamics satisfies the state, effort, and terminal constraints. Otherwise, fallback guidance is used. Algorithm~\ref{alg:closed_loop} summarizes the recursion, and Fig.~\ref{fig:framework} shows one HoloOcean replan.

\subsection{Conformal Terminal Prediction}
\label{subsec:terminal_method}
For each deterministic replanning index $k<T$, let $\Pred_{T|k}$ be a terminal-location predictor fixed before deployment. Given candidate tracking history $Y_{0:k}^j$, it produces $\widehat y_{T|k}^j=\Pred_{T|k}(Y_{0:k}^j)$. An independent calibration set of $n_{\rm tr}$ trajectories provides the residuals
\begin{equation}
R_{T|k}^{(i)}=\norm{y_T^{(i)}-\Pred_{T|k}(Y_{0:k}^{(i)})}_2,\quad i=1,\ldots,n_{\rm tr}.
\label{eq:terminal_residual}
\end{equation}

Let $\ell_{\rm tr} = \left\lceil (n_{\rm tr}+1)(1-\alpha_{\rm tr}) \right\rceil$. Write the ordered calibration residuals as $R_{T|k,(1)} \le \cdots \le R_{T|k,(n_{\rm tr})}$, and augment them by $R_{T|k,(n_{\rm tr}+1)}=+\infty$. We define $q_{T|k}=R_{T|k,(\ell_{\rm tr})}$. Thus, when $\ell_{\rm tr}=n_{\rm tr}+1$, the conformal radius is infinite and the resulting terminal region is $\mathbb R^p$. The terminal region for candidate $j$ is
\[
\mathcal Z_k^j=\left\{y\in\R^p:\norm{y-\widehat y_{T|k}^j}_2\le q_{T|k}\right\}.
\]
The regions define the assignment reach cost and branch terminal constraints. Theorem~\ref{thm:terminal_coverage} gives fixed-time marginal coverage of the true target's terminal position at any prespecified replan.

\subsection{Querying and Online Identity Update}
\label{subsec:identity_method}
Let $g_\theta(m,z)\in[0,1]$ be a pre-trained identity predictor, where $z$ is the pre-measurement sensing context. The target-class nonconformity score is $s(m,z)=1-g_\theta(m,z)$. A conditional-quantile model $\widehat q(z)$ is fitted on a separate training split and frozen~\cite{romano2019cqr}, giving the adjusted residual $\eta(m,z)=s(m,z)-\widehat q(z)$. This adjustment reduces systematic context-dependent score shifts, but does not by itself establish exchangeability. The anytime-validity result separately requires the target residual sequence to satisfy Assumption~\ref{assume:identity_exchange}.

A disjoint held-out target-class split is randomly permuted and partitioned into candidate-specific reference sequences
\[
\mathcal R^j=\{\eta_{{\rm ref},1}^j,\ldots,\eta_{{\rm ref},n_{\rm ref}}^j\},\quad j\in\Jset.
\]
Each candidate is tested as the target using a separate private reference sequence, avoiding reuse of raw reference residuals across streams. At replan $k$, a capacity-constrained nearest-center query partition $\{\widetilde I_{a,k}\}_{a\in\Aset}$ of $\Cset_k$ makes each retained candidate query-eligible for at most one agent. This partition is used only for query selection.

Let $v_k^{a,j}\ge0$ be a sensing-quality score and let $\mathcal Q_{a,k}\subseteq\widetilde I_{a,k}$ contain the candidates admissible to agent $a$. The evidence gap is $g_k^j=\log(1/\alpha_{\rm id})-\log(\max\{\overline W_k^j,\eps\})$. The balanced-evidence policy prioritizes, in order, the fewest previous queries, smallest evidence gap, highest sensing quality, and smallest candidate index. Each agent selects at most one admissible candidate and issues no query when $\mathcal Q_{a,k}$ is empty. Selection may use controller-visible history and policy randomization, but not the current identity measurement or private residuals.

Let $n_k^j$ denote the number of queries of candidate $j$ before replan $k$. Suppose candidate $j$ is selected for query $r=n_k^j+1$ and produces identity measurement $m_k^{a,j}$ under context $z_k^{a,j}$. The identity update computes $\eta_r^j=\eta(m_k^{a,j},z_k^{a,j})$ and appends it to candidate $j$'s sequence. Let $N_r=n_{\rm ref}+r$ and write $(\xi_1^j,\ldots,\xi_{N_r}^j)=(\eta_{{\rm ref},1}^j,\ldots,\eta_{{\rm ref},n_{\rm ref}}^j,\eta_1^j,\ldots,\eta_r^j)$. The smoothed online conformal p-value is
\begin{equation}
p_r^j=\frac{\sum_{i=1}^{N_r}\1\{\xi_i^j>\xi_{N_r}^j\}+U_r^j\sum_{i=1}^{N_r}\1\{\xi_i^j=\xi_{N_r}^j\}}{N_r},
\label{eq:identity_pvalue}
\end{equation}
where $U_r^j\sim\operatorname{Unif}(0,1)$ is independent. Large target-class nonconformity therefore produces a small p-value. Under the true-target hypothesis, the conditions in Section~\ref{sec:guarantees} make these p-values conditionally uniform relative to the controller-visible pre-query history~\cite{vovk2005algorithmic,fedorova2012plugin}.

For fixed $\gamma\in(0,1)$, define $e_r^j=\gamma(p_r^j)^{\gamma-1}$ and $E_k^j=e_{n_k^j+1}^j$ if candidate $j$ is queried at replan $k$, and $E_k^j=1$ otherwise. Starting from $W_0^j=\overline W_0^j=1$, update
\begin{equation}
W_{k+1}^j=W_k^jE_k^j,\quad \overline W_{k+1}^j=\max\{\overline W_k^j,W_{k+1}^j\}.
\label{eq:eprocess}
\end{equation}
Large e-process values are evidence against candidate $j$ being the target. The retained set is $\Cset_k=\{j:\overline W_k^j<1/\alpha_{\rm id}\}$ and $\widehat\Cset_k=\Cset_{k+1}$. The running maximum makes rejection permanent, while Theorem~\ref{thm:anytime_retention} gives anytime-valid target retention. The set $\widehat\Cset_k$ is passed to assignment and DDTO.

\begin{algorithm}[t]
\caption{Closed-loop multi-agent DDTO controller}
\label{alg:closed_loop}
\footnotesize
\begin{algorithmic}[1]
\STATE Initialize $\Cset_0=\Jset$, $W_0^j=\overline W_0^j=1$, $n_0^j=0$, and $B_{a,\rm rem}(0)=B_a$.
\FOR{$k=0,\ldots,T-1$}
 \STATE Construct $\{\mathcal Z_k^j\}$, select identity queries, and update
 $\widehat\Cset_k$ using~\eqref{eq:identity_pvalue}--\eqref{eq:eprocess}.
 \STATE Form the terminal sets and solve the pruned reach-cost assignment; if infeasible, form nearest-distance fallback groups.
 \STATE If the assignment is feasible, solve~\eqref{eq:ddto_problem} with~\eqref{eq:common_prefix} for each agent with $I_{a,k}\neq\emptyset$ over $r\in\mathcal R_{a,k}(I_{a,k})$.
 \STATE For each agent, apply the first control from a verified DDTO plan when available, otherwise apply fallback guidance.
 \STATE Update states, histories, and remaining budgets.
\ENDFOR
\end{algorithmic}
\end{algorithm}

\subsection{Assignment, DDTO, and Applied Control}
\label{subsec:planning_method}

Let $\chi_{a,j,k}=1$ when candidate $j\in\widehat\Cset_k$ is assigned to agent $a$, and define $I_{a,k}=\{j:\chi_{a,j,k}=1\}$. Each retained candidate is assigned exactly once and each agent receives at most $\kappa_a$ candidates:
\begin{equation}
\begin{aligned}
\sum_{a\in\Aset}\chi_{a,j,k}&=1,
&&j\in\widehat\Cset_k,\\
\sum_{j\in\widehat\Cset_k}\chi_{a,j,k}&\le\kappa_a,
&&a\in\Aset
\end{aligned}
\label{eq:assignment_constraints}
\end{equation}

Following the reference-branch DDTO construction in~\cite{elango2022deferring,elango2025deferred}, consider agent $a$ with an assigned candidate set $I$. DDTO constructs one state-control branch $\{x_{h|k}^{a,j},u_{h|k}^{a,j}\}$ for each $j\in I$. One candidate $r\in I$ is used as the reference branch. For every $j\in I\setminus\{r\}$ and $h=0,\ldots,H_k$, let $\delta_{h|k}^{a,j}\in\{0,1\}$ indicate whether branch $j$ has separated from the reference branch by stage $h$. When $\delta_{h|k}^{a,j}=0$, the branches coincide. When $\delta_{h|k}^{a,j}=1$, they may differ. The resulting DDTO common-prefix constraints are
\begin{subequations}
\label{eq:common_prefix}
\begin{align}
\delta_{0|k}^{a,j}=0, \quad \delta_{h+1|k}^{a,j} &\ge \delta_{h|k}^{a,j},
&&h=0,\ldots,H_k-1,\\
\norm{x_{h|k}^{a,j}-x_{h|k}^{a,r}}_2
&\le M_x\delta_{h|k}^{a,j},
&&h=0,\ldots,H_k,\\
\norm{u_{h|k}^{a,j}-u_{h|k}^{a,r}}_2
&\le M_u\delta_{h|k}^{a,j},
&&h=0,\ldots,H_k-1.
\end{align}
\end{subequations}
These constraints hold for every $j\in I\setminus\{r\}$. The monotonicity constraint makes separation irreversible over the planning horizon. The condition $\delta_{0|k}^{a,j}=0$ forces all branches to share the control applied at the current replan. The constants $M_x$ and $M_u$ bound feasible differences after separation.

Before the final replan, each branch terminates in its candidate's conformal region. At $k=T-1$, when $q_{T|T-1}\le\rho$, a tightened set is used:
\begin{align}
\mathcal T_{a,k}^{j,\rm reg}
&=\{x\in\Xset_a:P_ax\in\mathcal Z_k^j\},\notag\\
\mathcal T_{a,T-1}^{j,\rm cap}
&=\left\{x\in\Xset_a:
\norm{P_ax-\widehat y_{T|T-1}^j}_2
\le\rho-q_{T|T-1}\right\}.
\label{eq:robust_capture_set}
\end{align}
Under terminal coverage, membership in the tightened set implies terminal distance at most $\rho$. Thus,
\begin{equation}
\mathcal X_{a,k}^{j,\rm term}
=
\begin{cases}
\mathcal T_{a,T-1}^{j,\rm cap},
&k=T-1,\ q_{T|T-1}\le\rho,\\
\mathcal T_{a,k}^{j,\rm reg},
&\text{otherwise}.
\end{cases}
\label{eq:planning_terminal_set}
\end{equation}

For a fixed assigned set $I$ and reference branch $r$, we solve the DDTO problem of~\cite{elango2022deferring,elango2025deferred},
\begin{subequations}
\label{eq:ddto_problem}
\begin{align}
\Psi_{a,k}^r(I):={}&
\minimize_{\{x,u,\delta\}}
\sum_{j\in I\setminus\{r\}}
\sum_{h=0}^{H_k-1}\delta_{h|k}^{a,j}\\
\st \quad
x_{h+1|k}^{a,j}
&=f_a(x_{h|k}^{a,j},u_{h|k}^{a,j}),
\quad h=0,\ldots,H_k-1,\\
x_{h|k}^{a,j}&\in\Xset_a,
\quad h=0,\ldots,H_k,\\
u_{h|k}^{a,j}&\in\Uset_a,
\quad h=0,\ldots,H_k-1,\\
\sum_{h=0}^{H_k-1}
&c_a(x_{h|k}^{a,j},u_{h|k}^{a,j})
\le B_{a,\rm rem}(k),\\
x_{H_k|k}^{a,j}
&\in\mathcal X_{a,k}^{j,\rm term}, \quad x_{0|k}^{a,j}=x_k^a.
\end{align}
\end{subequations}
All constraints apply to each $j\in I$ together with \eqref{eq:common_prefix}. The objective delays separation while preserving terminal feasibility.

Let $\mathcal R_{a,k}(I)\subseteq I$ denote the reference candidates evaluated by the DDTO solver. For $r\in\mathcal R_{a,k}(I)$, we solve~\eqref{eq:ddto_problem} and select the feasible reference with the smallest objective value. If none is feasible, the DDTO solve is declared infeasible. In the numerical implementation, we evaluate at most two references, ordered by increasing $\overline W_{k+1}^j$ and then terminal-center distance.

The controller minimizes total assignment cost subject to~\eqref{eq:assignment_constraints}. Let $r_{j,k}^{\rm term}$ be the position-space radius of \eqref{eq:planning_terminal_set}. The distance from agent $a$ to the terminal set of candidate $j$ is
\[
d_{a,j,k}
=
\pospart{
\norm{P_ax_k^a-\widehat y_{T|k}^j}_2
-r_{j,k}^{\rm term}
}.
\]
We use $c_{a,j,k}^{\rm asg}=(1+\lambda_e g_{k+1}^j)d_{a,j,k}$ and prune candidate-agent pairs that fail a coarse remaining-horizon reachability screen. The subsequent DDTO solve and nonlinear rollout verification enforce dynamic, budget, and terminal feasibility.

When the pruned reach-cost assignment is feasible, its sets $I_{a,k}$ define the fallback groups. Otherwise, a capacity-constrained nearest-distance partition defines them. An agent without a verified DDTO plan applies fallback guidance toward the candidate in its fallback group with smallest $\overline W_{k+1}^j$, with ties broken by terminal-center distance and index. Clipping and scaling enforce input and remaining-budget constraints.
After applying controls, $B_{a,\rm rem}(k+1)=B_{a,\rm rem}(k)-c_a(x_k^a,u_k^a)$, and the updated state, budget, histories, and identity evidence initialize the next replan.

\section{Theoretical Guarantees}
\label{sec:guarantees}

We establish fixed-time terminal coverage, anytime-valid target retention, and a verified final-step capture bound. We condition on the trained models, reference allocation, and query-rule parameters. Controller randomization remains part of the information filtration. Let $(\Omega,\mathcal F,\Prob)$ support all problem data and random variables. For $\mathsf H_j:=\{J^\star=j\}$, write $\Prob_j(\cdot)=\Prob(\cdot\mid\mathsf H_j)$ and $\E_j$ for expectation under $\Prob_j$.

The pre-query sigma-field $\mathcal H_k^-$ contains states, budgets, tracking histories, past query decisions and conformal outputs, retained sets, previous plans and controls, and the randomization used for the replan-$k$ queries. It excludes candidate-specific reference residuals and raw identity measurements. Terminal regions, the query partition, query choices, and sensing contexts are therefore $\mathcal H_k^-$-measurable. After the selected measurements are converted internally into conformal p-values, define
\[
 \mathcal H_k^+=\mathcal H_k^-\vee
 \sigma\!\left(\{p_{n_k^j+1}^j:j\text{ is queried at }k\}\right).
\]
The evidence factors, $W_{k+1}^j$, $\widehat\Cset_k$, assignment, branches, and controls are $\mathcal H_k^+$-measurable, and $\mathcal H_k^+\subseteq\mathcal H_{k+1}^-$. Thus, $\Cset_k$ is the pre-query set and $\widehat\Cset_k=\Cset_{k+1}$ is used for planning at time $k$.

\begin{assumption}
\label{assume:terminal_exchange}
For every deterministic $k<T$ and every $j$ with $\Prob(J^\star=j)>0$, conditional on the frozen predictor $\Pred_{T|k}$, the calibration residuals in \eqref{eq:terminal_residual} and $R_{T|k}^{\star,j}=\norm{y_T^j-\Pred_{T|k}(Y_{0:k}^j)}_2$ are exchangeable under $\mathsf H_j$.
\end{assumption}
Assumption~\ref{assume:terminal_exchange} gives the exchangeability condition required for fixed-time terminal coverage.

\begin{theorem}
\label{thm:terminal_coverage}
Under Assumption~\ref{assume:terminal_exchange}, for every $k<T$,
\begin{equation}
 \Prob\!\left(y_T^{J^\star}\in\mathcal Z_k^{J^\star}\right)
 \ge 1-\alpha_{\rm tr}.
 \label{eq:terminal_coverage}
\end{equation}
This is a fixed-time marginal statement over calibration and test randomness.
\end{theorem}

\begin{proof}
Fix $j$ and condition on $\mathsf H_j$. The exchangeable-rank argument applied to the augmented order statistic, with a weak inequality and therefore conservative handling of ties, gives $\Prob_j(R_{T|k}^{\star,j}\le q_{T|k})\ge1-\alpha_{\rm tr}$. This event is $y_T^j\in\mathcal Z_k^j$. Averaging over $J^\star$ proves \eqref{eq:terminal_coverage}.
\end{proof}

This establishes fixed-time terminal-location coverage but does not control errors accumulated through repeated identity rejection. We therefore turn to validity under the adaptive query history.
For candidate $j$, let $\mathcal G_{r-1}^j$ denote the controller-visible information immediately before its $r$-th query. This field contains its previous conformal p-values, the controller-visible outputs of the other candidate streams, tracked geometry, previous plans and controls, and policy randomization. It excludes candidate $j$'s raw reference and runtime residuals and its unused tie-breaking variables.

\begin{assumption}
\label{assume:identity_exchange}
For each candidate $j$, let $\mathcal E^j$ contain the exogenous scene, other candidate streams, and controller-visible randomization, but not candidate $j$'s raw residuals or tie-breaking variables. Under $\mathsf H_j$, conditional on $\mathcal E^j$, the reference and potential runtime residuals are exchangeable. The variables $U_1^j,\ldots,U_T^j$ are conditionally independent $\operatorname{Unif}(0,1)$ variables and are independent of these residuals given $\mathcal E^j$. Finally, the $r$-th query predictably reveals the next residual and satisfies $\mathcal G_{r-1}^j\subseteq \mathcal E^j\vee\sigma(p_1^j,\ldots,p_{r-1}^j)$.
\end{assumption}

Thus, predictable querying reveals the next exchangeable residual rather than selecting among unseen residuals, yielding the following conditional rank-validity result.

\begin{lemma}
\label{lem:online_rank_validity}
Under Assumption~\ref{assume:identity_exchange}, for every candidate $j$, query index $r=1,\ldots,T$, and $u\in[0,1]$,
\[
\Prob_j\!\left(p_r^j\le u\mid\mathcal G_{r-1}^j\right)=u.
\]
In particular, $p_r^j$ is conditionally superuniform relative to the complete controller-visible pre-query history.
\end{lemma}

\begin{proof}
Fix $j$ and condition on $\mathcal E^j$. The initial reference block and the potential runtime residuals then form one exchangeable sequence. With independent randomized tie breaking, the rank of each newly revealed residual among the reference block and all residuals revealed so far is uniform and independent of the preceding randomized ranks~\cite{vovk2005algorithmic,fedorova2012plugin}. Predictability ensures that the $r$-th query reveals the next residual rather than selecting among unseen values. Therefore $\Prob_j\!\left(p_r^j\le u \mid \mathcal E^j,p_1^j,\ldots,p_{r-1}^j\right)=u$.
The inclusion in Assumption~\ref{assume:identity_exchange} and the tower property give the claim.
\end{proof}

\begin{theorem}
\label{thm:anytime_retention}
Under Assumption~\ref{assume:identity_exchange},
\begin{equation}
\Prob\!\left(J^\star\in\Cset_k\ \text{for every }k=0,\ldots,T\right)
\ge1-\alpha_{\rm id}.
\label{eq:anytime_retention}
\end{equation}
\end{theorem}

\begin{proof}
Under $\mathsf H_j$, Lemma~\ref{lem:online_rank_validity} gives $\E_j[e_r^j\mid\mathcal G_{r-1}^j]\le1$, while unqueried steps contribute one. Hence $\{W_k^j\}$ is a nonnegative supermartingale with $W_0^j=1$. Ville's inequality gives $\Prob_j(\max_{k\le T}W_k^j\ge1/\alpha_{\rm id})\le\alpha_{\rm id}$.
Averaging over the true index proves~\eqref{eq:anytime_retention}.
\end{proof}

A sufficient condition for Assumption~\ref{assume:identity_exchange} is that, conditional on $\mathcal E^j$, the adjusted target residuals are i.i.d. and unaffected by time, querying agent, geometry, or previous controls. Conditional independence across candidate streams, together with this within-stream condition, is sufficient. Non-target residuals may follow a different law.

We next combine target retention, terminal coverage, and verified feasibility for the regular terminal-set case. Corollary~\ref{cor:capture} handles the tightened final-replan case. For deterministic $k<T$, let $F_k$ denote the event that the assignment and required DDTO branches are feasible and pass nonlinear rollout verification. Let $R_k$ denote the event that the regular terminal set is used. By \eqref{eq:planning_terminal_set}, it holds for $k<T-1$ and equals $\{q_{T|T-1}>\rho\}$ at $k=T-1$. Let $B_k^{\rm reg}$ denote the event that the target has an assigned branch satisfying $P_a x_{H_k|k}^{a,J^\star}\in\mathcal Z_k^{J^\star}$.

\begin{theorem}
\label{thm:integrated}
Fix deterministic $k<T$. Under Assumption~\ref{assume:identity_exchange},
\begin{equation}
 \Prob(F_k\cap R_k\cap(B_k^{\rm reg})^c)
 \le\alpha_{\rm id}.
 \label{eq:branch_bound}
\end{equation}
Under Assumptions~\ref{assume:terminal_exchange} and~\ref{assume:identity_exchange},
\begin{equation}
 \Prob\!\left(F_k\cap R_k\cap
 \left(B_k^{\rm reg}\cap
 \{y_T^{J^\star}\in\mathcal Z_k^{J^\star}\}\right)^c\right)
 \le\alpha_{\rm id}+\alpha_{\rm tr}.
 \label{eq:branch_coverage_bound}
\end{equation}
On $F_k$, each agent with a nonempty assignment applies the common first input of its returned branches. Both bounds hold at any prespecified replan $k<T$.
\end{theorem}

\begin{proof}
Theorem~\ref{thm:anytime_retention} gives $\Prob(J^\star\notin\widehat\Cset_k)\le\alpha_{\rm id}$. On the complement, $F_k\cap R_k$ assigns the target exactly once and enforces its terminal-region constraint, so $F_k\cap R_k\cap \{J^\star\in\widehat\Cset_k\} \subseteq B_k^{\rm reg}$.
This proves \eqref{eq:branch_bound}. Theorem~\ref{thm:terminal_coverage} and a union bound give \eqref{eq:branch_coverage_bound}. Finally, $\delta_{0|k}^{a,j}=0$ in \eqref{eq:common_prefix} forces all first controls assigned to agent $a$ to agree.
\end{proof}

\begin{figure*}[t]
\centering
\includegraphics[width=0.95\textwidth]{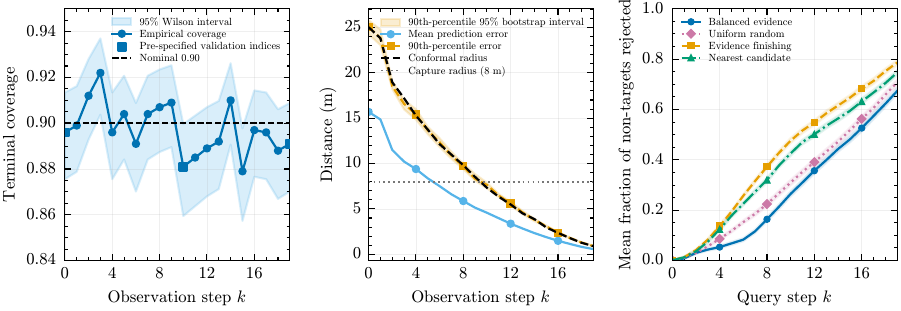}
\caption{Statistical component validation over 1000 Monte Carlo trials. Terminal coverage, prediction error and conformal radius, and mean non-target rejection are reported. Shading gives pointwise 95\% intervals.}
\label{fig:statistical-diagnostics}
\end{figure*}

At $k=T-1$, the tightened terminal set $\mathcal T_{a,T-1}^{j,\rm cap}$ in \eqref{eq:robust_capture_set} is used whenever $q_{T|T-1}\le\rho$, leading to the following result.
\begin{corollary}
\label{cor:capture}
Let $F_{T-1}^{\rm cap}$ be the event that $q_{T|T-1}\le\rho$, the tightened terminal sets are defined as in~\eqref{eq:robust_capture_set}, the returned assignment and DDTO problems are feasible, and nonlinear rollout under the planning dynamics $f_a$ verifies the resulting control sequence. Under Assumptions~\ref{assume:terminal_exchange} and~\ref{assume:identity_exchange},
\[
 \Prob(F_{T-1}^{\rm cap}\cap\mathrm{Cap}^c)
 \le\alpha_{\rm id}+\alpha_{\rm tr}.
\]
\end{corollary}

\begin{proof}
On target retention and $F_{T-1}^{\rm cap}$, the target is assigned to a verified branch satisfying~\eqref{eq:robust_capture_set}. Because one step remains, its common first input is the applied input. On terminal coverage, the triangle inequality gives $\norm{P_ax_T^a-y_T^{J^\star}}_2 \le (\rho-q_{T|T-1})+q_{T|T-1}=\rho$. A union bound over target rejection and terminal miscoverage proves the claim.
\end{proof}
Thus, Corollary~\ref{cor:capture} bounds capture failure only when a verified tightened-set plan is available at the final replan. It does not bound unconditional failure of the complete controller, whose fallback cases are evaluated empirically in Section~\ref{sec:experiments}.

\section{Numerical Evaluation}
\label{sec:experiments}

We now test closed-loop capture and control effort against CMPC, DDMP, commitment, and fallback baselines under a strict control-effort budget, then repeat the same paired scenarios after relaxing that budget. Second, we use component and fallback ablations to separate the effects of candidate rejection, query selection, assignment, and fallback guidance. Finally, we report planning feasibility and computation.

\subsection{Setup and Evaluation Protocol}
\label{subsec:experiment_setup}

We use two agents, eight candidates, $T=20$, $\Delta t=1$~s, $\kappa_a=4$, $\rho=8$~m, $M_x=300$, and $M_u=900$. The experiments are stop-the-clock: $\Delta t$ is the simulated dynamics/control discretization, not a wall-clock replanning period, and optimization time is not charged to the mission. At each replan, both agent states and all candidate tracking histories are available to the centralized controller without communication delay or loss. The same tracking and identity measurement models are used for both agents. We evaluate 100 independently generated scenarios shared across methods at per-agent budgets $B_a\in\{150,250\}\times10^3$~N$^2$s. The candidate layout and budgets were selected in preliminary sensitivity experiments and then frozen.

\begin{table}[t]
\centering
\caption{Closed-loop comparison at $B_a=150\times10^3$~N$^2$s over 100 paired scenarios. Effort is reported as median [IQR]. Bold denotes the highest observed capture rate.}
\label{tab:budget-150}
\begin{tabular}{@{}lcc@{}}
\toprule
Method & Capture (\%)$\uparrow$ &
Team effort ($10^3$ N$^2$s)$\downarrow$ \\
\midrule
Proposed DDTO & \textbf{79.0\%} & 214.9 [164.1, 274.5] \\
CMPC~\cite{cmpc} & 67.0\% & 233.4 [168.6, 281.8] \\
DDMP~\cite{ddmp} & 46.0\% & 257.3 [187.4, 300.0] \\
\midrule
Fallback-only & 6.0\% & 300.0 [300.0, 300.0] \\
Fallback-disabled & 13.0\% & 46.6 [34.4, 65.6] \\
\midrule
Center--commit & 13.0\% & 217.2 [171.9, 300.0] \\
Immediate--commit & 53.0\% & 189.0 [155.2, 223.8] \\
\bottomrule
\end{tabular}
\end{table}

\begin{table*}[t]
\centering
\caption{Component ablations at $B_a=150\times10^3$~N$^2$s over 100 paired scenarios. Effort and complete verified-plan rate are reported as median [IQR]. Bold denotes the highest observed capture rate.}
\label{tab:closed-loop-ablation}
\footnotesize
\resizebox{\textwidth}{!}{%
\begin{tabular}{llllccc}
\toprule
Variant & Candidate rejection & Query rule & Assignment &
\shortstack{Capture $\uparrow$} &
Effort ($10^3$ N$^2$s)$\downarrow$ &
Complete verified-plan rate (\%)$\uparrow$ \\
\midrule
Proposed DDTO
& Yes
& Balanced evidence
& Reach cost
& \textbf{79.0\%}
& 214.9 [164.1, 274.5]
& 80.0 [80.0, 85.0] \\

No candidate rejection
& No
& Balanced evidence
& Reach cost
& 68.0\%
& 268.1 [234.1, 300.0]
& 80.0 [75.0, 80.0] \\

Uniform-random query
& Yes
& Uniform random
& Reach cost
& 67.0\%
& 172.7 [121.3, 229.6]
& 85.0 [80.0, 86.0] \\

Clustered assignment
& Yes
& Balanced evidence
& Terminal clustering
& 63.0\%
& 260.5 [209.2, 300.0]
& 80.0 [80.0, 85.0] \\
\bottomrule
\end{tabular}}
\end{table*}

The Monte Carlo comparisons use the reduced six-degree-of-freedom Fossen model~\cite{fossen2011handbook}, $M\dot{\nu}=\tau-D\nu$, with HoveringAUV mass and actuation parameters~\cite{potokar2024holoocean}. We use $\operatorname{diag}(M)=(31.03,31.03,31.03,0.182,6.690,6.690)$ and $\operatorname{diag}(D)=(1.552,1.552,1.552,0.071,1.153,6.690)$, with body-force and body-speed norm bounds of $400$~N and $10$~m/s, respectively, depth limits $[-15,-1]$~m, and RK4 integration. Thruster dynamics and control allocation are omitted.

We separately demonstrate closed-loop execution with the HoloOcean~\cite{potokar2024holoocean} HoveringAUV plant. Simulator feedback initializes each replan, while DDTO optimizes with the reduced model and applies only its first control before receiving new simulator feedback. This is a qualitative execution demonstration under model mismatch, and is not covered by Corollary~\ref{cor:capture}. Figure~\ref{fig:framework} shows a snapshot at one replan, while the supplementary video presents a separate complete closed-loop episode.

The agents start at $(\pm20,-40,-5.5)$~m with zero velocity. We initialize the eight candidates in a sparse-maximin horizontal layout within the $75$~m simulation domain. The fixed anchor pattern is randomly permuted across candidate identities and perturbed by $0.5$~m horizontal Gaussian jitter. Initial candidate depth is centered at $-5.5$~m and forward body velocity is centered at $1$~m/s. Candidates follow the same vehicle dynamics under bounded AR(1) body forces and are constrained to depths in $[-12,-2]$~m.

Tracking measurements are corrupted by i.i.d. Gaussian noise with covariance $0.25^2I_3$~m$^2$. We use a constant-velocity predictor for terminal prediction and calibrate the conformal radius separately at each replan using 2000 trajectories.

Identity logits follow $\mathcal N(1,1)$ for the target and $\mathcal N(1-1.875v,1)$ for non-targets, where $v=\exp(-r_k^{a,j}/75)\operatorname{clip}(d_{{\rm iso},k}^j/10,0.35,1.35)$, $r_k^{a,j}$ is agent--candidate range, and $d_{{\rm iso},k}^j$ is the nearest-candidate separation. Thus, target logits do not depend on the sensing context in this study. The experiment does not test validity under context-dependent target-score shifts. Queries are permitted within $60$~m. Each candidate has 250 private target-class reference residuals, with $(\alpha_{\rm tr},\alpha_{\rm id},\gamma)=(0.10,0.05,0.55)$ and $\lambda_e=0.2$. Assignment uses a remaining-horizon reachability screen with $2$~m slack.

The primary closed-loop configuration uses two sequential convex programming (SCP) iterations, each linearizing the RK4-discretized dynamics by finite differences about the current nominal trajectory and updating the nominal with the accepted solution. We use fixed $\ell_\infty$ trust regions of $100$ for state and $400$~N for input and solve each convexified problem with MOSEK under a $60$~s time limit. A solution is accepted only when its maximum constraint and integrality violations are at most $10^{-4}$. A nonlinear RK4 rollout of the returned controls is then checked against the state, effort, and terminal constraints. Plans failing these checks, including those affected by solver timeouts, are treated as unverified, and the controller applies fallback guidance.

A shifted verified-plan tail initializes SCP when available. Otherwise, fallback-guidance rollouts are used. We use $c_a(x,u)=\Delta t\norm{u}_2^2$ and compare per-agent budgets $B_a\in\{150,250\}\times10^3$~N$^2$s, with the same limit applied to each branch-planning problem.

All candidate-contingent methods use the same guidance rollout and fallback controller: a saturated proportional--derivative (PD) law toward the predicted terminal center with gains $26$~N/m and $34$~N\,s/m. World-frame position and velocity errors generate a body-frame force, saturated at $400$~N and scaled when needed to satisfy remaining-budget constraints.

\subsection{Compared Controllers}

The CMPC~\cite{cmpc} and DDMP~\cite{ddmp} baselines use the same sensing, identity update, terminal regions, assignment, dynamics, fallback, effort constraints, and replanning procedure as Proposed DDTO. CMPC couples only the first control across branches, whereas DDMP fixes a three-stage common prefix. Proposed DDTO instead minimizes the branch-separation objective in~\eqref{eq:ddto_problem}.

The clustered-assignment variant sorts retained candidates by the first coordinate of their predicted terminal centers and partitions the ordered list into contiguous capacity-feasible assigned sets. Immediate--commit selects, for each agent, the candidate in its assigned set with the smallest $\overline W_{k+1}^j$, retains that choice until the candidate is rejected, and solves a singleton terminal trajectory problem at every replan. This gives a simpler committed single-trajectory baseline. Center--commit plans toward the centroid of the assigned candidates' predicted terminal centers until $\lceil0.95T\rceil$, after which it selects the candidate with the smallest $\overline W_{k+1}^j$.

Fallback-only uses the same sensing, candidate rejection, reach-cost assignment, and fallback guidance policy as Proposed DDTO, but applies fallback guidance at every replan without solving a DDTO problem. Fallback-disabled retains DDTO planning but applies zero body force whenever no verified DDTO plan is available.

\begin{figure}[t]
\centering
\includegraphics[width=0.85\columnwidth]{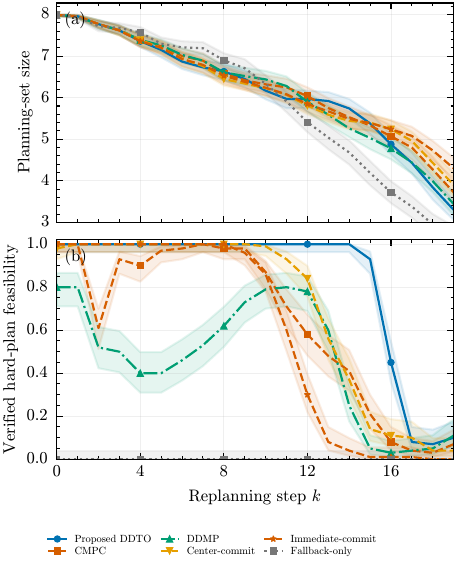}
\caption{Planning diagnostics under the relaxed $250\times10^3$~N$^2$s per-agent budget over 100 paired scenarios: (a) post-query retained-set size and (b) complete verified-plan rate for the compared controllers. Shading denotes pointwise 95\% intervals.}
\label{fig:planning-diagnostics}
\end{figure}

\subsection{Resource-Constrained Performance}

We separately use 1000 Monte Carlo trials to validate the statistical components in Fig.~\ref{fig:statistical-diagnostics}. The balanced-evidence policy retains the true target throughout the horizon in $99.3\%$ of trials, corresponding to an empirical target-rejection rate of $0.7\%$ over the horizon, below the nominal $\alpha_{\rm id}=5\%$ level. The figure also reports terminal coverage, terminal-region contraction, and non-target rejection.

Tables~\ref{tab:budget-150} and~\ref{tab:budget-250} show that Proposed DDTO attains the highest observed capture rate among the primary planners under both budgets: $79\%$ versus $67\%$ CMPC and $46\%$ DDMP under the strict budget, and $85\%$ versus $74\%$ for both baselines under the relaxed budget.

\begin{table}[t]
\centering
\caption{Closed-loop comparison after relaxing the per-agent effort budget to $B_a=250\times10^3$~N$^2$s over the same 100 paired scenarios. Effort is reported as median [IQR]. Bold denotes the highest observed capture rate.}
\label{tab:budget-250}
\begin{tabular}{@{}lcc@{}}
\toprule
Method & Capture (\%)$\uparrow$ &
Team effort ($10^3$ N$^2$s)$\downarrow$ \\
\midrule
Proposed DDTO & \textbf{85.0\%} & 271.3 [181.4, 339.3] \\
CMPC~\cite{cmpc} & 74.0\% & 253.2 [162.4, 357.4] \\
DDMP~\cite{ddmp} & 74.0\% & 334.7 [246.4, 428.1] \\
\midrule
Fallback-only & 28.0\% & 500.0 [477.5, 500.0] \\
Fallback-disabled & 14.0\% & 46.8 [34.9, 67.0] \\
\midrule
Center--commit & 12.0\% & 288.2 [208.2, 380.6] \\
Immediate--commit & 52.0\% & 189.4 [151.0, 232.1] \\
\bottomrule
\end{tabular}
\end{table}

Table~\ref{tab:closed-loop-ablation} evaluates the component ablations in the strict-budget setting, where preserving alternatives is most constrained by the effort budget. Removing candidate rejection, using uniform-random queries, and replacing reach-cost assignment reduce observed capture rate from $79\%$ to $68\%$, $67\%$, and $63\%$, respectively. Effort is reported as median [IQR]. Uniform-random querying uses less effort, indicating an observed capture--effort tradeoff.

\subsection{Planning Feasibility and Computation}

Under the relaxed budget, Proposed DDTO has a median complete verified-plan rate of $80.0\%$ across the 100 scenarios, as shown in Fig.~\ref{fig:planning-diagnostics}. Fallback-disabled captures $13/100$ targets under the strict budget and $14/100$ after the budget is relaxed, while Fallback-only captures $6/100$ and $28/100$, respectively. Proposed DDTO captures $79/100$ and $85/100$. Neither Fallback-disabled nor Fallback-only reproduces the capture rate of Proposed DDTO.

One fixed-reference problem $\Psi_{a,k}^r(I)$ has $O(|I_{a,k}|H_k(n_a+m_a))$ continuous and $O((|I_{a,k}|-1)H_k)$ binary variables. Monotone separation gives at most $(H_k+1)^{|I_{a,k}|-1}$ separation patterns per reference. We evaluate at most two references with two SCP iterations per agent.

We measure computation over 50 Proposed DDTO scenarios under the relaxed budget on an Apple M4 Pro MacBook. In the primary closed-loop implementation, the two agent problems are solved in parallel, while up to two reference problems and two SCP iterations are evaluated sequentially within each agent. The median and 95th-percentile end-to-end replan wall times are $5.99$ and $16.12$~s, respectively. The fixed-reference DDTO problems are independent and can be evaluated in parallel. On the 50-scenario benchmark, parallel reference evaluation reduces the median and 95th-percentile end-to-end replan times to $4.15$ and $11.49$~s, respectively, without changing the DDTO formulation.

We additionally evaluated a one-iteration SCP configuration with parallel reference evaluation on the same 100 scenarios used for the relaxed-budget two-iteration SCP closed-loop study. Its median and 95th-percentile replan wall times were $2.41$ and $6.76$~s. In this configuration, an individual fixed-reference solve of $\Psi_{a,k}^r(I)$ had a median solve time of $1.83$~s, indicating that the fixed-reference DDTO solves dominate computation. On the same scenarios, one-iteration SCP captured $90/100$ targets versus $85/100$ for two-iteration SCP while reducing wall-clock time. These stop-the-clock timings do not demonstrate $1$~s real-time replanning.


\section{Conclusion}
\label{sec:conclusion}
We presented a multi-agent DDTO framework for interception under uncertain target identity and terminal location. The analysis establishes fixed-time terminal-location coverage, anytime-valid target retention, and a bound on the probability that a verified tightened-set plan is returned at the final replan but capture fails. Proposed DDTO achieves the highest observed capture rates among the planners: $79\%$ and $85\%$ under strict and relaxed effort budgets, respectively. Fallback ablations show that neither Fallback-disabled nor Fallback-only reproduces the performance of the complete controller.


\bibliographystyle{IEEEtran}
\bibliography{ref}

@INPROCEEDINGS{cmpc,
  author={Alsterda, John P. and Brown, Matthew and Gerdes, J. Christian},
  booktitle={2019 American Control Conference (ACC)}, 
  title={Contingency Model Predictive Control for Automated Vehicles}, 
  year={2019},
  pages={717--722},
  doi={10.23919/ACC.2019.8815260}}

@article{potokar2024holoocean,
  title={{HoloOcean}: A full-featured marine robotics simulator for perception and autonomy},
  author={Potokar, Easton and Lay, Kalliyan and Norman, Kalin and Benham, Derek and Ashford, Spencer and Peirce, Randy and Neilsen, Tracianne B and Kaess, Michael and Mangelson, Joshua G},
  journal={IEEE Journal of Oceanic Engineering},
  volume={49},
  number={4},
  pages={1322--1336},
  year={2024},
  publisher={IEEE},
  doi={10.1109/JOE.2024.3410290}
}

@INPROCEEDINGS{ddmp,
  author={Isele, David and Añon, Alexandre Miranda and Tariq, Faizan M. and Yeh, Goro and Singh, Avinash and Bae, Sangjae},
  booktitle={2025 IEEE International Conference on Robotics and Automation (ICRA)}, 
  title={Delayed-Decision Motion Planning in the Presence of Multiple Predictions}, 
  year={2025},
  pages={11743--11749},
  doi={10.1109/ICRA55743.2025.11128178}}

@book{fossen2011handbook,
  title={Handbook of marine craft hydrodynamics and motion control},
  author={Fossen, Thor I},
  year={2011},
  publisher={John Wiley \& Sons}
}

@article{mesforush2026,
  title={Lasso Universality Under Linearly Dependent Covariates in the Sparse Regime},
  author={Mesforush, Soroush and Parhi, Rahul},
  journal={arXiv preprint arXiv:2608.08390},
  year={2026}
}

@inproceedings{norouzi2025novel,
  title={A novel adaptive formation control strategy for teams of unmanned vehicles under complete dynamic uncertainty},
  author={Norouzi, Maryam and Zhou, Mingxi and Yuan, Chengzhi},
  booktitle={AIAA AVIATION FORUM AND ASCEND 2025},
  pages={3192},
  year={2025},
  doi={10.2514/6.2025-3192}
}

@article{lindemann2023safe,
  author = {Lindemann, Lars and Cleaveland, Matthew and Shim, Gihyun and Pappas, George J.},
  title = {Safe Planning in Dynamic Environments Using Conformal Prediction},
  journal = {IEEE Robotics and Automation Letters},
  volume = {8},
  number = {8},
  pages = {5116--5123},
  year = {2023},
  doi = {10.1109/LRA.2023.3292071}
}

@article{sheng2024safe,
  title={Safe {POMDP} online planning among dynamic agents via adaptive conformal prediction},
  author={Sheng, Shili and Yu, Pian and Parker, David and Kwiatkowska, Marta and Feng, Lu},
  journal={IEEE Robotics and Automation Letters},
  volume={9},
  number={11},
  pages={9946--9953},
  year={2024},
  publisher={IEEE}
}

@article{strawn2023conformal,
  title={Conformal predictive safety filter for {RL} controllers in dynamic environments},
  author={Strawn, Kegan J and Ayanian, Nora and Lindemann, Lars},
  journal={IEEE Robotics and Automation Letters},
  volume={8},
  number={11},
  pages={7833--7840},
  year={2023},
  publisher={IEEE}
}

@article{lee2025dmvc,
  title={{DMVC-Tracker}: Distributed Multi-Agent Trajectory Planning for Target Tracking Using Dynamic Buffered Voronoi and Inter-Visibility Cells},
  author={Lee, Yunwoo and Park, Jungwon and Kim, H Jin},
  journal={IEEE Robotics and Automation Letters},
  year={2025},
  volume={10}, number={5}, pages={4842--4849},
  publisher={IEEE},
  doi={10.1109/LRA.2025.3551255}
}

@article{khass2026multi,
  title={Multi-Agent Next-Best-View Optimization for Risk-Averse Planning},
  author={Khass, Amirhossein Mollaei and Pandey, Vivek and Liu, Guangyi and Cosse, Athanasios and Bayrak, Emrah and Motee, Nader},
  journal={arXiv preprint arXiv:2606.04158},
  year={2026}
}

@article{yuan2025hierarchical,
  title={Hierarchical Planning for Long-Horizon Multi-Target Tracking Under Target Motion Uncertainty},
  author={Yuan, Junbin and Moon, Brady and Cao, Muqing and Scherer, Sebastian},
  journal={IEEE Robotics and Automation Letters},
  year={2025},
  publisher={IEEE},
  volume={10}, number={12}, pages={13113--13120},
  doi={10.1109/LRA.2025.3625492}
}

@article{peters2024contingency,
  title={Contingency games for multi-agent interaction},
  author={Peters, Lasse and Bajcsy, Andrea and Chiu, Chih-Yuan and Fridovich-Keil, David and Laine, Forrest and Ferranti, Laura and Alonso-Mora, Javier},
  journal={IEEE Robotics and Automation Letters},
  volume={9},
  number={3},
  pages={2208--2215},
  year={2024},
  publisher={IEEE}
}

@article{de2021scenario,
  title={Scenario-based trajectory optimization in uncertain dynamic environments},
  author={de Groot, Oscar and Brito, Bruno and Ferranti, Laura and Gavrila, Dariu and Alonso-Mora, Javier},
  journal={IEEE Robotics and Automation Letters},
  volume={6},
  number={3},
  pages={5389--5396},
  year={2021},
  publisher={IEEE}
}

@ARTICLE{calafiore2012robust,
  author={Calafiore, Giuseppe C. and Fagiano, Lorenzo},
  journal={IEEE Transactions on Automatic Control}, 
  title={Robust Model Predictive Control via Scenario Optimization}, 
  year={2013},
  volume={58},
  number={1},
  pages={219-224},
  doi={10.1109/TAC.2012.2203054}}

@inproceedings{dixit2023adaptive,
  title={Adaptive conformal prediction for motion planning among dynamic agents},
  author={Dixit, Anushri and Lindemann, Lars and Wei, Skylar X and Cleaveland, Matthew and Pappas, George J and Burdick, Joel W},
  booktitle={Learning for Dynamics and Control Conference},
  pages={300--314},
  year={2023},
  organization={PMLR}
}

@article{fusioncp,
title = {Conformal prediction-based multi-sensor classification},
journal = {Neurocomputing},
volume = {678},
pages = {133182},
year = {2026},
issn = {0925-2312},
doi = {10.1016/j.neucom.2026.133182},
author = {Farbod Siahkali and Vijay Gupta},
}

@ARTICLE{siahkali2026active,
  author={Siahkali, Farbod and Hou, Mengxue and Gupta, Vijay},
  journal={IEEE Control Systems Letters}, 
  title={Active Sensing and Deferred-Decision Trajectory Optimization for Robust Target Identification}, 
  year={2026},
  volume={10},
  pages={811--816},
  doi={10.1109/LCSYS.2026.3703893}}

@inproceedings{elango2022deferring,
  author = {Elango, Purnanand and Sarsilmaz, Selahattin Burak and Acikmese, Behcet},
  title = {Deferring Decision in Multi-Target Trajectory Optimization},
  booktitle = {AIAA SciTech 2022 Forum},
  pages = {1583},
  year = {2022},
  doi = {10.2514/6.2022-1583}
}

@article{elango2025deferred,
  author = {Elango, Purnanand and Sarsilmaz, Selahattin Burak and Acikmese, Behcet},
  title = {Deferred-Decision Trajectory Optimization},
  journal = {arXiv preprint arXiv:2502.06623},
  year = {2025},
  eprint = {2502.06623},
  archivePrefix = {arXiv},
  primaryClass = {math.OC}
}

@book{vovk2005algorithmic,
  author = {Vovk, Vladimir and Gammerman, Alexander and Shafer, Glenn},
  title = {Algorithmic Learning in a Random World},
  publisher = {Springer},
  address = {New York},
  year = {2005}
}

@article{angelopoulos2023gentle,
  author = {Angelopoulos, Anastasios N. and Bates, Stephen},
  title = {Conformal Prediction: A Gentle Introduction},
  journal = {Foundations and Trends in Machine Learning},
  volume = {16},
  number = {4},
  pages = {494--591},
  year = {2023},
  doi = {10.1561/2200000101}
}

@inproceedings{romano2019cqr,
  author = {Romano, Yaniv and Patterson, Evan and Cand{\`e}s, Emmanuel},
  title = {Conformalized Quantile Regression},
  booktitle = {Advances in Neural Information Processing Systems},
  volume = {32},
  year = {2019}
}

@inproceedings{muthali2023reachcal,
  author = {Muthali, Anish and Shen, Haotian and Deglurkar, Sampada and Lim, Michael H. and Roelofs, Rebecca and Faust, Aleksandra and Tomlin, Claire J.},
  title = {Multi-Agent Reachability Calibration with Conformal Prediction},
  booktitle = {2023 62nd IEEE Conference on Decision and Control (CDC)},
  pages = {6596--6603},
  year = {2023},
  doi = {10.1109/CDC49753.2023.10383723}
}

@article{vovk2021evalues,
  author = {Vovk, Vladimir and Wang, Ruodu},
  title = {E-Values: Calibration, Combination, and Applications},
  journal = {The Annals of Statistics},
  volume = {49},
  number = {3},
  pages = {1736--1754},
  year = {2021},
  doi = {10.1214/20-AOS2020}
}

@book{ville1939etude,
  author = {Ville, Jean},
  title = {{\'E}tude Critique de la Notion de Collectif},
  series = {Th{\`e}ses de l'entre-deux-guerres},
  number = {218},
  publisher = {Gauthier-Villars},
  address = {Paris},
  year = {1939}
}

@inproceedings{fedorova2012plugin,
  author={Valentina Fedorova and Alex Gammerman and Ilia Nouretdinov and Vladimir Vovk},
  title={Plug-in Martingales for Testing Exchangeability On-line},
  booktitle={Proceedings of the 29th International Conference on Machine Learning},
  pages={1639--1646},
  year={2012}
}

@inproceedings{rezaei2025online,
  title={An Online Path Planner for Bandwidth-aware Aerial Camera Networks},
  author={Rezaei, Mohammad Amin and Manfredi, Gioacchino and Racanelli, Vito Andrea and De Cicco, Luca and Mascolo, Saverio},
  booktitle={2025 17th International Congress on Ultra Modern Telecommunications and Control Systems and Workshops (ICUMT)},
  pages={300--305},
  year={2025},
  organization={IEEE},
  doi={10.1109/ICUMT67815.2025.11268743}
}

@article{mazloom2026barl,
title = {{BARL}: Batch active regression learning via query-by-committee with application to slope factor of safety prediction},
journal = {Applied Soft Computing},
volume = {188},
pages = {114412},
year = {2026},
issn = {1568-4946},
doi = {10.1016/j.asoc.2025.114412},
author = {Shayan Mazloom and Abdollah Shafieezadeh and Jieun Hur and Jae-Wook Jung and Jeong-Gon Ha and Daegi Hahm},
}
\end{document}